\documentclass{article}
\usepackage[utf8]{inputenc}
\usepackage[T1]{fontenc}
\usepackage[english]{babel}

\usepackage{blindtext}
\usepackage{amssymb,amsmath,amsthm}
\usepackage[linesnumbered,algoruled,boxed,lined]{algorithm2e}
\usepackage{enumerate} 
\usepackage{a4wide}
\usepackage{yfonts}
\allowdisplaybreaks
\usepackage[active]{srcltx}
\usepackage{comment}
\usepackage{epsfig}
\usepackage{xspace}
\usepackage{fancyhdr}
\usepackage{url}
\usepackage{hyperref}
\usepackage{color}
\usepackage[pass]{geometry}
\usepackage[normalem]{ulem}
\usepackage{authblk}

\usepackage{todonotes}

\newtheorem{theorem}{Theorem}[section]

\newtheorem{lemma}[theorem]{Lemma}
\newtheorem{proposition}[theorem]{Proposition}
\newtheorem{claim}[theorem]{Claim}

\theoremstyle{remark}

\theoremstyle{definition}

\providecommand{\keywords}[1]
{
  \small	
  \textbf{\textit{Keywords---}} #1
}

\newcommand{\rv}[1]{{\color{black}{#1}}}

\title{On the proper interval completion problem within some chordal subclasses}

\date{}

\newcommand{\nina}[1]{\todo[fancyline, color=orange!20]{{\bf Nina: }  #1}\xspace}

\newcommand{\valeria}[1]{\todo[fancyline, color=red!10]{{\bf Valeria: }  #1}\xspace}

\newcommand{\fran}[1]{\todo[fancyline, color=green!10]{{\bf François: }  #1}\xspace}

\newcommand{\suggestion}[2]{#2}
\newcommand{\newsuggestion}[2]{#2}

\author[1]{François Dross}
\author[1]{Claire Hilaire}
\author[2]{Ivo Koch}
\author[3]{Valeria Leoni}
\author[4,5]{Nina Pardal}
\author[3]{María Inés Lopez Pujato}
\author[6]{Vinicius Fernandes dos Santos}
\affil[1]{Univ. Bordeaux, CNRS, Bordeaux INP, LaBRI, UMR 5800, F-33400 Talence, France}
\affil[2]{University of General Sarmiento, Argentina}
\affil[3]{National University of Rosario and CONICET, Argentina}
\affil[4]{University of Buenos Aires, Argentina}
\affil[5]{ICC-CONICET, Argentina}
\affil[6]{Computer Science Department, Federal University of Minas Gerais, Brazil}

\begin{document}

\maketitle

\begin{abstract}
Given a property (graph class) $\Pi$, a graph $G$, and an integer $k$, the \emph{$\Pi$-completion} problem consists of deciding whether we can turn $G$ into a graph with the property $\Pi$ by adding at most $k$ edges to $G$. 
 The $\Pi$-completion problem is known to be NP-hard for general graphs when $\Pi$ is the property of being a proper interval graph (PIG). 
 In this work, we study the PIG-completion problem 
 within different subclasses of chordal graphs. We show that the problem remains NP-complete even when restricted to split graphs. We then turn our attention to positive results and present polynomial time algorithms to solve the PIG-completion problem when the input is restricted to caterpillar and threshold graphs. We also present an efficient algorithm for the minimum co-bipartite-completion for quasi-threshold graphs, which provides a lower bound for the PIG-completion problem within this graph class.
\end{abstract}


\section{Introduction}









Graph modification problems can be used to address many fundamental problems, not only in graph theory itself, but also to model a large number of practical applications in several different fields. Some of \suggestion{the fields in which graph modification problems are widely used to model applications are}{of those fields include} molecular biology, computational algebra, and more \suggestion{precisely}{generally}, \suggestion{in}{} areas that involve modelling based on graphs where the missing edges are due to a lack of data, for example in data clustering problems~\cite{Goldberg1995,Natanzon2001}. In many of these applications, an edge modification of the graph that models the experimental data corresponds to correcting errors and inconsistencies in the data. 

Given a graph $G$ and a graph property $\Pi$, a \emph{graph modification problem} consists of studying how to add or delete the minimum number of vertices or edges from $G$ in order to obtain a graph that satisfies the property $\Pi$.
In this article, the property $\Pi$ will represent a graph class, such that the graph resulting from the modification belongs to this class. We focus our attention on one of the four basic graph modification problems: the $\Pi$-completion problem.

Given a property $\Pi$, a \emph{$\Pi$-completion} of a graph $G = (V, E)$ is a supergraph $H = (V, E \cup F)$ such that $H$ belongs to $\Pi$ and $E \cap F = \emptyset$. The edges in $F$ are referred to as \emph{fill edges}. A $\Pi$-completion $H = (V, E \cup F)$ of $G$ is \emph{minimum} if, for any set of fill edges $F'$ such that $H'= (V, E \cup F')$ belongs to $\Pi$, it holds that $|F'| \geq |F|$. In this case, $|F|$ is called the \emph{$\Pi$-completion number} of $G$.
The \emph{minimum $\Pi$-completion problem} consists in 
finding the $\Pi$-completion number of a graph $G$. The associated decision problem--the \emph{$\Pi$-completion problem}--consists in deciding, for a given integer $k$, if $G$ has a $\Pi$-completion with at most $k$ fill edges. Throughout this work, a $\Pi$-completion will always be a minimum one unless otherwise stated.

The $\Pi$-completion problem from an arbitrary graph is known to be NP-complete when $\Pi$ is the class of chordal, interval, or proper interval graphs~\cite{GAREY, Goldberg1995, Kashiwabara1979, Yannakakis1981}. 
Furthermore, when $\Pi$ is the class of interval graphs, it was shown that the problem remains NP-complete on line graphs~\cite{GAREY}, and also on co-bipartite graphs~\cite{Yuan1998}. \suggestion{Díaz et al.\ showed that it can be solved in $O(n)$-time for trees~\cite{Diaz1991}.}{According to Peng et al.~\cite{PengChen06}, the problem can be solved in $O(n)$-time for trees, following from the results in~\cite{Diaz1991}}. 
In the case of chordal-completion and proper interval-completion (\emph{PIG}-completion from now on), its study from the viewpoint of parameterized complexity was initiated by Kaplan et al.\ in 1999~\cite{K-S-T-interval}. In 2015, Bliznets et al.\ presented the first subexponential parameterized algorithm for PIG-completion that finds a solution in $k^{O(k^{2/3})}+O(nm(kn+m))$-time~\cite{Bliznets2015}.\\
\indent Monadic second order logic for graphs is a fragment of second order logic in which formulas allow logical operations ($\vee, \wedge, \neg, \Rightarrow$), adjacency/incidency tests, membership tests ($v \in V, e \in E)$, quantification over vertices or edges ($\exists v \in V, \exists e \in E, \forall v \in V, \forall e \in E)$, and over sets of vertices or edges ($\exists U \subseteq V, \forall D \subseteq E$). When quantification over edge sets is not allowed, the logic is called $MSOL_1$, whereas when it is allowed it is called $MSOL_2$.
A celebrated result by Courcelle et al.\ states that each graph property that is expressible in $MSOL_1$ (resp. $MSOL_2$) can be solved in polynomial time for graphs with bounded cliquewidth (resp. treewidth)~\cite{Courcelle00lineartime}. Note that this result is mainly of theoretical interest and does not lead to practical algorithms. Since the problem of \suggestion{}{finding a} PIG-completion \suggestion{}{with at most $k$ edges} can be expressed in $MSOL_2$ \suggestion{}{for fixed $k$}, this motivates our search of efficient algorithms for subclasses of chordal graphs with bounded treewidth.\\
\indent Another direction \suggestion{in}{on} which current research on this topic is focused is in finding and characterizing minimal completions 
for input graph classes for which the minimum version is hard in the most efficient possible way from a computational point of view~\cite{CrespelleT13,HeggernesSTV07,RapaportST08}. 


Throughout this work, we consider the target class $\Pi$ to be the subclass of interval graphs given by proper interval graphs. 
The most well known motivation for the 
PIG-completion problem comes from molecular biology. In~\cite{Benzer1959}, Benzer first gave strong evidence that the collection of DNA composing a bacterial gene was linear. This linear structure could be represented as overlapping intervals on the
real line, and therefore as an interval graph. In order to study various
properties of a certain DNA sequence, the original piece of DNA is fragmented into smaller pieces which are then cloned many times. When
all the clones have the same size, the resulting graph should not only be interval, but proper interval.
Deciding whether two clones should overlap or not is the critical part. However, there might be some false positive or false negatives due to erroneous interpretation of some data. Thus, correcting the model to get rid of inconsistencies is equivalent to removing or adding as few edges as possible to the graph so that it becomes interval. 




This work is organized as follows. We start by proving in Section~\ref{sec:split_PIG} that the PIG-completion problem remains hard even if the input graph $G$ is split. Since split graphs are in particular chordal, this result implies that the problem remains hard when $G$ is a chordal graph, which leads us to study proper subclasses of chordal graphs where the PIG-completion problem might be tractable.
More precisely, in Section~\ref{sec:threshold} we first give an efficient algorithm for the PIG-completion problem on threshold graphs. We finish the section by showing an efficient dynamic programming algorithm for the co-bipartite-completion problem on quasi-threshold graphs
following a brief discussion on the difficulties of generalizing the previous result to this superclass of threshold graphs. In Section~\ref{sec:caterpillar}, we show an efficient algorithm for a very sparse class of graphs, a subclass of both interval graphs and trees called caterpillars.
We conclude the paper with some final remarks and possible future directions in Section~\ref{sec:conclusions}.


\section{Definitions}



We give in this section the basic definitions and fix the notation that will be used throughout this work. All graphs in this paper are undirected and simple. Let $G$ be a graph, and let $V(G)$ and $E(G)$ denote its vertex and edge sets, respectively. We denote by $n$ the number of vertices and by $m$, the number of edges.
Whenever it is clear from the context, we simply write $V$ and $E$ and denote $G=(V,E)$. \suggestion{}{For basic definitions not included here, we refer the reader to \cite{Bondy}}. 

Given a graph $G$ and $S\subseteq V$, the \emph{subgraph of $G$ induced by $S$}, denoted by $G[S]$, is the graph with vertex set $S$ and such that two vertices of $S$ are adjacent if and only if they are adjacent in $G$. When \suggestion{$G'=G[S]$}{$G'$ and $G[S]$ are isomorphic} for some $S \subseteq V$, \suggestion{}{with a slight abuse of terminology} we simply say that $G'$ is an \emph{induced subgraph} of $G$.
For any family $\mathcal{F}$ of graphs, we say that $G$ is \emph{$\mathcal{F}$-free} if $G$ does not contain any graph $F\in \mathcal{F}$ as an induced subgraph. If a graph $G$ is $\mathcal{F}$-free, then the graphs in $\mathcal{F}$ are called the \emph{forbidden induced subgraphs} of $G$.

A \emph{clique} \suggestion{of}{in} a graph $G$ is a complete induced subgraph of $G$. Also, we will often use this term for the vertex set that induces the clique.


\newsuggestion{}{A graph is \emph{chordal} if it does not admit an induced cycle of size $4$ or more.}

\newsuggestion{}{A graph $G$ is an \emph{interval} graph if it admits an intersection model consisting of intervals on the real line, that is, a family $\mathbf{I}$ of intervals on the real line and a mapping from the set of vertices of $G$ to the intervals of $\mathbf{I}$ such that two vertices are adjacent in $G$ if and only if the corresponding intervals intersect. Notice that the class of interval graphs is a subclass of the class of chordal graphs.}


A \emph{proper interval} graph is an interval graph that admits a \emph{proper interval model}, this is, an intersection model in which no interval is properly contained in any other. 

\newsuggestion{}{A \emph{unit interval} graph is an interval graph that has an interval representation in which each interval has unit length. Every proper interval graph is a unit interval graph, and viceversa \cite{roberts1969indifference}.}

Three nonadjacent vertices of a graph form an $AT$ (\emph{asteroidal triple}) if every two of them are connected by a path avoiding the neighbourhood of the third.
Interval graphs are precisely those chordal graphs that are also $AT$-free \newsuggestion{}{\cite{Lekkeikerker}} and proper interval graphs are precisely those chordal graphs that are also  $\{\text{claw}, \text{tent}, \text{net}\}$-free \newsuggestion{}{\cite{roberts_phd_th,roberts1969indifference}} (see Figure~\ref{fig:S3clawnet}).

\begin{figure}
    \centering
    \includegraphics{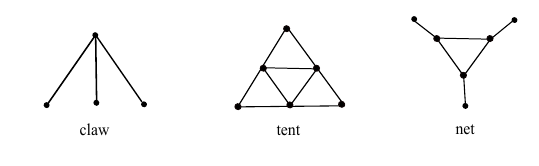}
    \caption{Forbidden induced subgraphs for proper interval graphs}
    \label{fig:S3clawnet}
\end{figure} 

A graph $G=(C \cup I,E)$ is a \emph{split graph} if its vertex set can be partitioned into a set $C$ of pairwise adjacent vertices and set $I$ of pairwise nonadjacent vertices.

A \emph{threshold} graph is a split graph in which any two \newsuggestion{independent}{nonadjacent} vertices satisfy that the neighborhood of one is contained in the neighborhood of the other.
Equivalently, $G$ is a threshold graph if it can be constructed from the empty graph by repeatedly adding either an \emph{isolated vertex} (nonadjacent to every other vertex) or a \emph{dominating vertex} (adjacent to every other vertex). \newsuggestion{}{Let the ordering of $V(G)$ according to this construction procedure be the \emph{threshold ordering}}. Threshold graphs are characterised precisely as \newsuggestion{those}{the} $\{2K_2, C_4, P_4\}$-free graphs.

\emph{Quasi-threshold graphs}, also called \emph{trivially perfect graphs}, are \suggestion{those}{the} $\{P_4,C_4\}$-free graphs. A connected quasi-threshold graph \newsuggestion{$G = (V, E(G))$}{$G = (V, E_G)$} admits a rooted tree \newsuggestion{$T = (V, E(T))$}{$T = (V, E_T)$} on the same vertex set $V$, rooted on a vertex $r$, such that \newsuggestion{$uv \in E(G)$}{$uv \in E_G$} if and only if there is a path in $T$ starting in $r$ containing both $u$ and $v$.





A graph $G$ is a \emph{caterpillar} if $G$ is a tree in which the removal of all the pendant vertices (i.e., the \emph{leaves}) results in a path (i.\ e.\ the \emph{spine} or \emph{central path}).






\newsuggestion{}{A family $\mathcal{S}$ of nonempty sets has the \emph{Helly property} if every nonempty subfamily of $\mathcal{S}$ of pairwise intersecting sets has a nonempty intersection. This property is also known as the \emph{2-Helly property}. For example, any family of nonempty pairwise-intersecting intervals in the real line has the 2-Helly property.}


\section{PIG-completion within split graphs}\label{sec:split_PIG}

We now devote our attention to the complexity of the PIG-completion problem when the input graph belongs to the class of split graphs. 
We start by citing some useful results and stating a few lemmas that characterise PIG-completions when the split partition fulfills certain properties. These results will be useful in Section~\ref{subsec:reduction_split}, where we give a reduction to the \newsuggestion{problem}{PIG-completion problem}.

Recall the following result:

\begin{theorem}[Peng et al.~\cite{PengChen06}]\label{theo:threshold_to_split}
The threshold-completion problem is NP-complete on split graphs. 
\end{theorem}


Let us consider a connected split graph $G=(C\cup I, E)$ such that $C$ is maximum, in the sense that there is no vertex in $I$ \newsuggestion{complete to $C$}{adjacent to every vertex of $C$}, and let $H=(C\cup I, E\cup F)$ be a PIG-completion of $G$. 

\begin{lemma}\label{lemma:1_split}
There exist both a partition $\{C',I'\}$ of the vertex set of $H$, where $I'\subseteq I$ and $C'$ is a clique with $C'\supseteq C$, and a partition $\{I_l,I_r\}$ of $I'$ such that $H[I_l]$ and $H[I_r]$ are both cliques.
\end{lemma}

\begin{proof}
\newsuggestion{}{Let $\{\mathbf{I}_u\}_{u\in V(H)}$ be a unit interval model for the graph $H$, which we know it exists since $H$ is proper interval. It follows from the 2-Helly property of intervals and the fact that $C$ is a clique in $H$ that there exists a real point $p$ in $\bigcap_{u \in C} \mathbf{I}_u$. For each vertex $v \in I$, if either the right endpoint of $I_v$ is to the left of $p-1$, or the left endpoint of $I_v$ is to the right of $p+1$, then $I_v$ does not intersect any interval in $\{I_u: u \in C\}$. We may conclude that $v$ is pairwise nonadjacent to every vertex of $C$ in $H$ and thus $v$ is an isolated vertex in $G$, which leads to a contradiction given that $G$ is connected. Let us now consider $I_l = \{u \in I : p-1 \in \mathbf{I}_u\}$ and  $I_r = \{u \in I : p+1 \in \mathbf{I}_u\}$, which are both cliques in $H$. Moreover, if we define $C'= C\cup I \setminus (I_l \cup I_r)$, then it is clear from our previous remark that, for every vertex $u \in  I \setminus (I_l \cup I_r)$, the unit interval $I_u$ is fully contained in the interval $[p- 1, p+1]$ and thus it contains the point $p$. Furthermore, $p$ is also contained in the interval corresponding to every vertex in $C$, and therefore $C'$ is a clique in $H$ as well.}
\end{proof}

\begin{lemma}\label{lemma:2_split}
Let $H=(C' \cup I', E\cup F)$ be a PIG-completion of $G$ and $\{I_l, I_r\}$ be the partition of $I'$ into cliques as in the previous lemma. 
Then, the graphs given by $H_l=(C'\cup I_l, E\cup F \setminus F_{I'})$ and $H_r=(C'\cup I_r, E\cup F \setminus F_{I'})$ are threshold graphs, where $F_{I'}= \{ uv \in F : u,v \in I' \}$.
\end{lemma}

\begin{proof}
If $H_l$ is not threshold, then there are vertices $v,w \in I_l$ (nonadjacent in $H_l$) and $x,y \in C'$ such that $vx, wy \in E\cup F \setminus F_{I'}$ and $vy, wx \not\in E\cup F \setminus F_{I'}$. Notice that $v$ and $w$ are adjacent in $H$ since they both lie in $I_l$, which is a clique. Therefore, we find an induced $C_4$ in $H$, which results in a contradiction \newsuggestion{for}{since} $H$ is a proper interval graph.
\end{proof}

This theorem follows directly from the previous lemmas.

\begin{theorem}\label{theo:1_split}
If $H=(C' \cup I', E\cup F)$ is a PIG-completion of $G$, then there is a partition $\{I_l, I_r\}$ of $I'$ such that $H_l, H_r$ are threshold graphs and $I_l,I_r$ are cliques in $H$.
\end{theorem}

Finally, we will need the following property of proper interval graphs. 
Let $G$ be an interval graph, and let $\mathbf{I}_G$ be an interval model of $G$. We say that a clique $C$ of $G$ is a \emph{first clique} of $G$ in $\mathbf{I}_G$ if there is a real point $p$ such that $p$ intersects every interval corresponding to the vertices of $C$, and \newsuggestion{all the others intervals}{every other interval} of $\mathbf{I}_G$ \newsuggestion{are}{lies} strictly on the same side of $p$.
Notice that such a clique always exists, since it can be found by restricting to the rightmost interval of the model.
 
\begin{lemma}\label{lem:gluingPIG}
    Let $G=(V,E)$ and $G'=(V',E')$ be two vertex-disjoint proper interval graphs, with respective proper interval models $\mathbf{I}_G$ and $\mathbf{I}_{G'}$, and let $C$ (resp.\ $C'$) be a first clique of $G$ in $\mathbf{I}_G$ (resp.\ $G'$ in $\mathbf{I}_{G'}$). Then, the graph $G''=(V\cup V' , E\cup E' \cup F)$ is also a proper interval graph, where $F$ consists of all \newsuggestion{}{the} possible edges between $C$ and $C'$.
\end{lemma}
\begin{proof}
    Let $\mathbf{I}_G$ be a proper interval model of $G$. 
    Up to \newsuggestion{taking the mirror of }{inverting right and left in} $\mathbf{I}_G$, we can assume that $C$ is on the right side of the model. Thus, $G$ admits a proper interval model $\mathbf{I}_G=\{\mathbf{I}_u\}_{u\in V} = \{[a_u,b_u]\}_{u\in V}$, with $a_u<b_u$, and there is a real point $p$ intersecting every interval of $C$ such that all the other intervals lie strictly to the left of $p$. 
    
    Let $k\geq 1$ be the size of $C$, let $n \geq k$ the size of $V$,
    and let $V=\{u_1, \dots, u_{n}\}$ be an ordering of the vertices such that $C=\{u_1, \dots, u_{k}\}$ and $b_{u_{k+1}} < \dots < b_{u_{n}} < p \leq b_{u_1} < \dots < b_{u_k}$.
    For each $i\in \{1,\dots, n\}$, let $c_{u_i} = p+1+\frac{i}{k}$ if $i\leq k$, and $c_{u_i}=b_{u_i}$ for the remaining vertices. 
    Observe that 
    $ p < p+1 \newsuggestion{\leq}{<} c_{u_1} < \dots < c_{u_k} = p+2$. Moreover, for each $i\in \{1,\dots, n\}$, it follows from the definition that $a_{u_i}<p$, and also $c_{u_i}<p$ for every $i>k$.

    Consider now the interval model given by $\{[a_u,c_u]\}_{u\in V}$.
    Note that, since we have not modified the ordering of the endpoints of the intervals, this is also a proper interval representation of $G$ where every interval that does not represent a vertex of $C$ lies strictly on the left of $p$, and all the remaining intervals are to the left of $p+2$.

    \medskip

    Consider now the proper interval model of $G'$ given by $\mathbf{I}_{G'}=\{\mathbf{I}'_v\}_{v\in V'} = \{[a_v,b_v]\}_{v\in V'}$, and let $p'$ be a real point intersecting $C'$ such that every other interval lies strictly to the \newsuggestion{left}{right }
    of $p'$.  
    Upon shifting all those intervals, we can assume that $p'=p+2$.

    Let $k'\geq 1$ be the size of $C'$, let $n' \geq k'$ the size of $V'$,
    and let $V'=\{v_1, \dots, v_{n'}\}$ be an ordering of the vertices such that $C'=\{v_1, \dots, v_{k'}\}$ and $a_{v_{k'}} < \dots < a_{v_1} \leq p' <  a_{v_{k'+1}} < \dots < a_{v_{n'}}$.
    For each $i\in \{1,\dots, n'\}$, let $c_{v_i} = p+1-\frac{i}{k'}$ if $i\leq k'$, and let $c_v=a_v$ for every other vertex.
    Observe that $p = c_{v_{k'}} < \dots < c_{v_1} < p+2$. Furthermore, for each $i\in \{1,\dots, n'\}$ it follows that $b_{v_i}>p+2$, and also $c_{v_i}>p+2$ for every $i>k'$.

    Consider the interval model given by $\{[c_v,b_v]\}_{v\in V'}$.
    Once more, since we have not modified the ordering of the endpoints of the intervals, this yields a proper interval representation of $G'$ such that all the intervals that do not represent a vertex of $C'$ lie strictly to the right of $p+2$, and all the intervals of $C'$ lie to the right of $p$.

    \medskip

    Finally, consider the interval model $\mathbf{I}$ given by $\{[a_u,c_u]\}_{u\in V} \cup \{[c_v,b_v]\}_{v\in V'}$,
    and let us see that $\mathbf{I}$ is indeed a proper interval model.
    Suppose that there are two intervals $[a,b]$ and $[a',b']$, representing $u$ and $u'$ respectively, such that $a<a'<b'<b$. By construction, one of them lies in $V$ and the other in $V'$.
    Suppose without loss of generality that $u\in V$ and $u'\in V'$. Hence, $b\leq p+2$ and \newsuggestion{$p+2 < a'$, which contradicts the fact that $a'<b$.}{$p+2 < b'$, which contradicts the fact that $b'<b$.} We reach an analogous contradiction if $u'\in V$ and $u\in V'$. 

    Let $G''=(V'',E'')$ be the proper interval graph corresponding to the proper interval model $\mathbf{I}$. By construction, it follows that $V\cup V'= V''$ and $E\cup E' \subseteq E''$.
    For every $u\in C$ and $v\in C'$, their corresponding intervals intersect in $p+1$, hence $G'$ contains all the edges between $C$ and $C'$. Let $F$ be said set of edges. 
    Let $u \in V\setminus C$, let $v \in V'$, and let $[a,b], [a',b']$ \newsuggestion{}{be} their respective intervals. By construction, $a<b<p \leq a'<b'$ and thus $u$ and $v$ must be nonadjacent. Similarly, if $u \in V$ and $v \in V'\setminus C$, then $a<b\leq p+2 < a'<b'$.
    Therefore, the only edges between $V$ and $V'$ are precisely those in $F$.

 \end{proof}

\subsection{NP-completeness}\label{subsec:reduction_split}

We are now ready to prove that obtaining a PIG-completion is still NP-complete when the input graph is split. In order to do this, we strongly rely on the previous lemmas and the fact that threshold-completion on split graphs is also NP-complete.

\begin{theorem}\label{theo:2_split}
 The PIG-completion problem is NP-complete on split graphs.
\end{theorem}

\begin{proof}
Given a completion of a split graph, it is easy to check in polynomial time \newsuggestion{it}{if} this is in fact a PIG-completion, hence the problem is in NP. 

We give a reduction from threshold-completion on split graphs. Let $(G,\ell)$ be an instance of threshold-completion on split graphs, where $G = (C \cup I, E)$ is a (connected, for simplicity) split graph on $n$ vertices.

Consider the graph $G'$ defined as follows. Let $G_1=(C_1\cup I_1,E_1)$ and $G_2=(C_2\cup I_2,E_2)$ be two copies of $G$. 
For each $i\in \{1,2\}$, we consider $G'_i = (C'_i \cup I_i, E'_i)$, the graph constructed from $G_i$ by connecting \newsuggestion{}{$2n^2$} new vertices to all the vertices of $G_i$. We denote by $V_i$ the vertex set of $G_i$ and by $V'_i$, the vertex set of $G'_i$, for each $i\in \{1,2\}$.
Finally, connect all the vertices of $C'_1$ and $C'_2$ into a clique $C'$. Let \newsuggestion{}{$G'=(V', E')$} be the resulting split graph on \newsuggestion{}{$2(2n^2+n)$} vertices, where $V'=C'\cup I_1 \cup I_2$.

\begin{figure}
    \centering
    \includegraphics[width=0.8\linewidth]{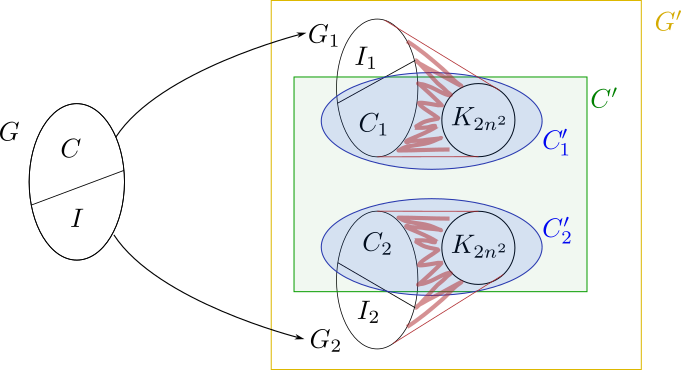}
    \caption{A schema of the gadget used for the reduction of Theorem \ref{theo:2_split}}
    \label{fig:my_label}
\end{figure}

We show that $G$ can be augmented to a threshold graph with at most $\ell$ \newsuggestion{}{fill} edges if and only if $G'$ can be augmented to a proper interval graph with at most $2k$ additional edges, where $k = \ell + \binom{|I|}{2}$.


First, suppose there is a minimum PIG-completion $H=(V' , E'\cup F)$ of $G'$ with $|F|\leq 2k$ edges. Notice that we may assume that $|F|= 2k$, since we can add additional edges if necessary and still keeping the property of being a PIG. One way to show this (assuming that all the endpoints in the interval model are distinct), is by shifting the leftmost interval to the right until it intersects a new interval.
We will show that $G$ can be augmented to a threshold graph with $\ell$ edges.

Observe that completing each $G_i$ into a clique requires less than $n^2$ edges. Thus, \newsuggestion{}{a PIG-completion of $G'$ would need less than $2n^2$ fill edges.} 

Let $\mathbf{I}_H=\{\mathbf{I}_u\}_{u\in V(H)}$ be a proper interval model for the graph $H$. 
By Theorem \ref{theo:1_split}, there is a partition $C'', I_l, I_r$ of $V'$, $C' \subseteq C''$, 
$I_l\cup I_r \subseteq I_1 \cup I_2$ such that $H[I_l]$ and $H[I_r]$ are both cliques with fill edges $F_{I'}$, and $H_l=(C''\cup I_l, E \cup F \setminus F_{I'})$ and $H_r=(C''\cup I_r, E \cup F \setminus F_{I'})$ are threshold graphs.

If there is a vertex $v \in C'' \setminus C'$, w.l.o.g. $v\in I_1$, then we need at least \newsuggestion{}{$2n^2$} fill edges to connect $v$ to $C'_2$ in $H$, which contradicts \newsuggestion{}{$|F|<2n^2$, thus $C''=C'$.} 

Since $C'$ induces a clique in $H$, it follows from the 2-Helly property of intervals that there is a real point $p\in \bigcap_{u \in C'} \mathbf{I}_u$. No other interval of the model $\mathbf{I}_H$ intersects the point $p$, otherwise this requires at least \newsuggestion{}{$2n^2$} fill edges, and thus \newsuggestion{the}{each} interval\newsuggestion{s}{} corresponding to the vertices in $I_1 \cup I_2 $ \newsuggestion{are}{lies} either strictly to the left or to the right of $p$.

\begin{claim} 
    There is no $v_1\in I_1$ and $v_2 \in I_2$  whose corresponding intervals \newsuggestion{}{lie either} both to the right of $p$ or both to the left of $p$. 
\end{claim}
\begin{proof}
    Let $\mathbf{I}_1$ and $\mathbf{I}_2$ be the intervals corresponding to $v_1$ and $v_2$, respectively. Suppose that $\mathbf{I}_1$ and $\mathbf{I}_2$ both lie to the left of $p$. Since neither $\mathbf{I}_1$ nor $\mathbf{I}_2$ intersects $p$, none of them appears to the right of $p$.
    Moreover, since $v_1\cup C'_1 \setminus C_1$ induces a clique in $H$, there is a real point $p_1$ intersecting all the intervals corresponding to this set, and the same holds for $v_2\cup C'_2 \setminus C_2$ and $p_2$. Suppose w.l.o.g. that $p_1 \leq p_2 < p$.
    Then, all the intervals of $C'_1 \setminus C_1$ intersect $p_2$, thus $v_2$ is adjacent to all the vertices of $C'_1 \setminus C_1$ introducing thus \newsuggestion{}{$2n^2$} fill edges, which results in a contradiction.

\end{proof}

Therefore, we assume without loss of generality that, in the interval representation of $H$, the endpoints of the intervals corresponding to vertices of $I_1$ lie strictly \newsuggestion{}{to} the left of $p$ and the \newsuggestion{}{ones corresponding to vertices} of $I_2$ lie strictly to the right \newsuggestion{}{of $p$}.
Since \newsuggestion{$I_l, I_r$}{$\{I_l, I_r\}$} is a partition of $I_1 \cup I_2$ where each set induces a clique in $H$, then $I_1=I_l$ and $I_2=I_r$.
Let us show that all the fill edges lie inside each $G_i$.

\begin{claim} 
    For each $i\in \{1,2\}$, let $F_i$ be the set of fill edges inside $H[V_i]$.
    Then $F = F_1 \cup F_2$.
\end{claim}
\begin{proof} 
    Recall that for each $i\in \{1,2\}$, all the possible edges between $V_i$ and $V'_i\setminus V_i$ are already in \newsuggestion{$E$}{$E'$}
    , thus $F_i$ is also the set of fill edges in $H[V'_i]$.

    Since we already proved that there is no fill edge between $I_1$ and $I_2$, it suffices to see that there is no fill edge between $I_1$ and $C'_2$ (resp.\ $I_2$ and $C'_1$).    
    Toward a contradiction, suppose there is at least one of said fill edges. Let us construct a PIG $H'$ on the same vertex set with edge set $E'\cup F_1 \cup F_2$.

    The model given by $\{\mathbf{I}_u\}_{u\in V'_1}$ is a proper interval model of $H[V'_1]$. It follows from the above reasoning that $C'_1$ can be seen as a first clique in this model, since the point $p$ lies exactly in the intervals that correspond to vertices of $C'$. Similarly, $\{\mathbf{I}_u\}_{u\in V'_2}$ is a proper interval model of $H[V'_2]$ and $C'_2$ is a first clique in this model.
    
    It follows from Lemma~\ref{lem:gluingPIG} that there is a PIG $H'=(V'_1\cup V'_2, E'_1\cup F_1 \cup E'_2 \cup F_2 \cup F')$, where $F'$ is the set of edges connecting all the vertices of $C'_1$ and $C'_2$.
    Notice that $E'_1 \cup E'_2 \cup F' = E'$, hence $H'$ is a PIG-completion of $G'$ with edge set $E'\cup F_1 \cup F_2$.
    Thus, $H'$ is a smaller PIG-completion than $H$, which contradicts the fact that $H$ has minimum number of edges.
\end{proof}

Assume without loss of generality that $|F_1| \le |F_2|$. Let us show that $|F_1| = |F_2|$. We will build a proper interval completion of $G'$ with $2|F_1|$ fill edges, which will then prove that $|F_1| = |F_2|$ since $H$ is a minimum PIG-completion of $G'$.
To do this, consider any proper interval model of $H$, and keep only those images of vertices belonging to $I_1 \cup C'_1$. Now cut the intervals right after point $p$, keeping arbitrarily small parts to the right of $p$ so that the intervals remain non-nested. By doing this we obtain a proper interval model of $H[I_1 \cup C'_1]$ where $C'_1$ is a first clique. Therefore, by considering twice the same graph in Lemma~\ref{lem:gluingPIG}, we get a proper interval graph that is exactly two copies of $H[I_1 \cup C'_1]$, that has all the edges between the two copies of $C'_1$. Note that since $G_1$ and $G_2$ are copies of the same graph, this proper interval graph is a PIG-completion of $G'$, with $2|F_1|$ fill edges, which proves that $|F_1| = |F_2| = k$.

Let $F'_1=\{uv \in F, u,v\in I_1\}$, and let $H_1=(V_1, E\cup F_1 \setminus F'_1)$. Note that since $H[I_1]$ is a clique and $G'[I_1]$ is an independent set, $|F'_1| = \binom{|I|}{2}$. In other words, all the possible edges between vertices of $I_1$ are fill edges. Therefore, $H_1$ has $\ell = k-\binom{|I|}{2}$ fill edges (as a completion of $G_1$). By Lemma~\ref{lemma:2_split}, $H_1$ must be a threshold-completion of $G_1$. Since $G_1$ is isomorphic to $G$, $G$ has a threshold-completion with $\ell$ edges.







\begin{figure}
    \centering
    \includegraphics[width=0.8\linewidth]{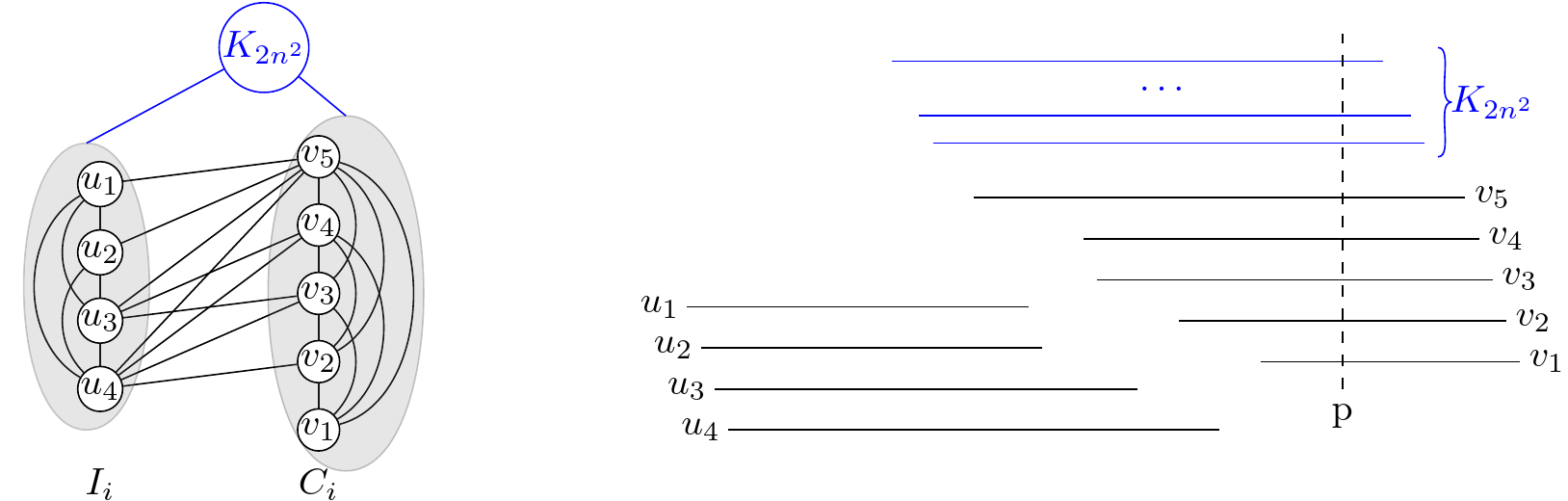}
    \caption{Example of $H''_i$ and its proper interval model with $n=9$.}
    \label{fig:exampleOnlyIf}
\end{figure}

For the only if direction, suppose there is a minimum threshold-completion $H$ of $G$ with \newsuggestion{$k'$}{$\ell$} fill edges. We will construct a PIG-completion of $G'$ with $k'$ fill edges such that $k' \newsuggestion{\leq}{=} 2\left(\ell + \binom{|I|}{2}\right)$. 

Let $F_i$ be the set of fill edges added to the vertices corresponding to each $G_i$ to obtain a threshold graph $H_i$ for each $i \in \{1,2\}$, and let $H'_i = (C'_i\cup I_i, E'_i\cup F_i)$. Notice that $H'_i$ is also a threshold-completion of $G'_i$ for each $i \in \{1,2\}$.
\newsuggestion{}{Observe that we can consider the same partition of the vertices into a clique and an independent set for both $G_i$ and $H_i$. Indeed, if a vertex $v\in I_i$ is in the clique of $H_i$, then we can remove all the edges from $v$ to vertices in the independent set. This way we could place $v$ in the independent set of $H_i$ instead. The same holds for the vertex partition of $G'_i$ and $H'_i$.}
Consider $F'_i$ to be the fill edges obtained by completing $I_i$ into a clique for each $i \in \{1,2\}$, and let $H''_i=(V'_i,E'\cup F_i \cup F'_i)$ for each $i \in \{1,2\}$. This gives a total of $\ell+\binom{|I|}{2}$ fill edges.




Since each $H'_i$ is a threshold graph, the neighbourhoods of the independent vertices are nested and hence we can consider an ordering of said vertices in terms of increasing containment of their neighbourhoods.
Recall that it is possible to represent any clique with a proper interval model by overlapping the corresponding intervals such that each interval starts and ends in a different point.
A proper interval model for $H''_i$ is given as follows, for each $i \in \{1,2\}$. Since $I_i$ is a clique in $H''_i$, then we can place the corresponding intervals such that they overlap. The same holds for the intervals corresponding to all the vertices in $C'_i$, and we can place the endpoints of the intervals corresponding to these vertices by following the ordering of the neighbourhoods of the vertices in $I_i$ to do this. This way, we can place the endpoints of the intervals corresponding to the vertices in $I_i$ according to the increasing ordering given by the neighbourhoods with regards to $C'_i$, and thus obtaining a proper interval model for each $H''_i$\newsuggestion{}{, as in Figure~\ref{fig:exampleOnlyIf}. Observe that $C'_i$ is a first clique in the described proper interval model of $H''_i$ for each $i \in \{1,2\}$}.
Finally, we obtain a PIG-completion of $G'$ with the desired number of fill edges by applying Lemma \ref{lem:gluingPIG}.
\end{proof}

\section{An algorithm for PIG-completion on threshold graphs}\label{sec:threshold}

\suggestion{We present in this section}{In this section we present} a simple linear-time algorithm for computing an optimal PIG-completion for a threshold graph $G$. To do this, we will show first that PIG-completion for threshold graphs is equivalent to co-bipartite-completion. This will enable us to give a procedure, based on the definition of threshold graphs, that iteratively places the vertices in one of the two cliques in an optimal way.
\newsuggestion{Let us}{We} consider the vertices in the \newsuggestion{}{threshold ordering} 
\newsuggestion{such that whenever a vertex is added, it is either isolated or dominating}{}.
For simplicity, we say that a vertex is \emph{dominating} if it is a dominating vertex \newsuggestion{when}{for the current iteration's graph at the moment it is} added, and \emph{isolated} if it is an isolated vertex \newsuggestion{when}{in the current iteration's graph when it is} added. Let us assume that the last vertex we add is dominating\suggestion{(otherwise the graph is not connected, and we apply our result on each connected component)}{, otherwise there are isolated vertices, which are irrelevant to the completion, and hence it is enough to solve the problem for the (single) nontrivial connected component, if it exists}.

\begin{lemma}\label{lemma:cobip_for_t}
\suggestion{}{Let $G$ be a threshold graph.} A minimum PIG-completion of $G$  is the same as a minimum co-bipartite-completion of $G$.
\end{lemma}

\begin{proof}
Since there is a dominating vertex \suggestion{}{in $G$} and the \suggestion{}{PIG-}completion is claw-free, its \suggestion{set of vertices}{vertex set} can be partitioned into two cliques. \suggestion{Assume}{Now, assume} that the vertex set of the completion can be partitioned into two cliques. \suggestion{There is obviously}{Thus, there is} no induced claw or asteroidal triple. Suppose there is an induced cycle $C$ of length at least four. Moreover, notice that $C$ has to be of size four since the graph is co-bipartite, and that two non-incident edges of $C$ are in $G$ (precisely the ones that are not inside a clique). Hence, every vertex in $V(C)$ has at least one neighbour and one non-neighbour in $G[V(C)]$, thus $G[V(C)]$ is not a threshold graph and neither is $G$, which results in a contradiction.
\end{proof}


Thus, it suffices to exhibit an algorithm that computes a minimum completion 
into two cliques, $C_1$ and $C_2$, and thus to exhibit their respective sets of vertices $S_1$ and $S_2$.

Let us consider the following algorithm:
We add the vertices in the given threshold ordering. When we add a vertex as a dominating vertex, we always put it in $S_1$. When we add a vertex as an isolated vertex, we compare the number of isolated vertices remaining to be added with the number of vertices already in $S_1$. If there are more remaining isolated vertices, then we put it in $S_1$. Otherwise, we put it in $S_2$. This very simple algorithm runs in time $O(n)$.
We will prove the following:

\begin{theorem} \label{th_algo_opti}
The previous algorithm gives an optimal partition $(S_1,S_2)$ of $V(G)$.
\end{theorem}

\begin{lemma}
The number of edges to add is the sum for each $i \in \{1,2\}$, for each isolated vertex $v$ that is added to $S_i$, of the number of vertices that are in $S_{i}$ at the time $v$ is added (without counting $v$).
\end{lemma}

\begin{proof}
One way to build the completion is to consider the vertices in order, and then, whenever we add an isolated vertex in $S_i$, add the edges between $v$ and all of the vertices in $S_i$.
\end{proof}

Let us now consider an optimal partition $(S_{o1},S_{o2})$.

\begin{lemma} \label{lem_domi1}
Let $v$ be a dominating vertex, and suppose that at least as many isolated vertices are placed in $S_{o2}$ after $v$ as in $S_{o1}$. Then, $v$ can be placed in $S_{o1}$ without worsening the solution.
\end{lemma}

\begin{proof}
For $i \in \{1,2\}$, let $I_i$ be the number of isolated vertices \suggestion{}{added} after $v$ in $S_{oi}$. If $v$ is in $S_{o2}$, then moving $v$ to $S_{o1}$ increases the solution cost by $I_1 - I_2 \le 0$.
\end{proof}

\begin{lemma} \label{lem_domi2}
Let $v$ be a dominating vertex, and suppose that when we add $v$ there are more vertices in $S_{o1}$ than in $S_{o2}$. Then, there is an optimal solution with the same partition up to $v$, such that $v$ is in $S_{o1}$.
\end{lemma}

\begin{proof}
For $i \in \{1,2\}$, let $k_i$ be the number of vertices in $S_{oi}$ that were added before $v$, and let $I_i$ be the number of isolated vertices that were added after $v$ in $S_{oi}$. Notice that $k_1 \ge k_2$.
If $I_1 \ge I_2$, then by swapping between the sets $S_{o1}$ and $S_{o2}$ all the vertices from $v$ on, we increase the solution cost by $(I_2 - I_1)(k_1 - k_2) \le 0$. \suggestion{Now}{Hence, we may assume w.l.o.g that} $I_1 \le I_2$, and thus by Lemma~\ref{lem_domi1} we know we can place $v$ in $S_{o1}$ if it is not there already.
\end{proof}

\begin{lemma} \label{lem_domi3}
There is an optimal solution such that every dominating vertex is put in the same clique, say $S_{o1}$. Moreover, in this solution there are always at least as many vertices in $S_{o1}$ as in $S_{o2}$.
\end{lemma}

\begin{proof}
Consider an optimal solution and its construction according to the ordering of the vertices. Let us transform this construction into another optimal solution.

We will do this maintaining that, in every step there is at least as many vertices in $S_{o1}$ as in $S_{o2}$, and that every dominating vertex is placed into $S_{o1}$. Whenever the solution places a dominating vertex in $S_{o2}$, we modify the solution by Lemma~\ref{lem_domi2}. If the solution places an isolated vertex $v$ in $S_{o2}$ while there are as many vertices in $S_{o1}$ as in $S_{o2}$, then swap all the remaining vertices (including $v$) between $S_{o1}$ and $S_{o2}$, thus we obtain a solution with the same cost.
\end{proof}

Let us now consider an optimal solution as in Lemma~\ref{lem_domi3}.

\begin{lemma} \label{lem_iso}
Let $v$ be an isolated vertex.
Suppose that every vertex before $v$ is in $S_{o1}$ and that there are at least as many remaining isolated vertices as there are vertices before $v$.
Then there is another optimal solution respecting Lemma~\ref{lem_domi3} with the same partition up to $v$ such that $v$ is in $S_{o1}$.
\end{lemma}

\begin{proof}
Assume that $v$ is in $S_{o2}$. For $i \in \{1,2\}$, let $I_i$ be the number of isolated vertices after $v$ in $S_{oi}$. 
Let $k$ be the number of vertices that are before $v$ in the order. Note that they are all in $S_{o1}$.

Moving $v$ from $S_{o2}$ to $S_{o1}$ increases the number of edges in the solution by $I_1 + k - I_2$. By hypothesis, $k \le (I_2 + I_1)$, so if $I_1 = 0$, then $I_1 + k - I_2 \le 0$, and we get another optimal solution respecting Lemma~\ref{lem_domi3}.

Now we may assume that $I_1 \ne 0$. Let $v'$ be the first isolated vertex put in $S_{o1}$ after $v$. For $i \in \{1,2\}$, let $k'_i$ be the number of vertices in $S_{oi}$ before $v'$, and let $I'_i$ be the number of isolated vertices after $v'$ in $S_{oi}$. Note that since our solution respects Lemma~\ref{lem_domi3}, every vertex in $S_{o2}$ is isolated, so $k_2' + I_2' = I_2+1$.

By putting $v$ in $S_{o1}$ and $v'$ in $S_{o2}$ we increase the cost of the solution by $m = I_1 + k - I_2 + (k_2' - 1 + I_2' - (k_1' + 1) - I_1')$. As noted previously, $k_2' + I_2' = I_2 + 1$, so $m = I_1 + k -k_1' - I_1' - 1$.
Moreover, as every isolated vertex after $v$ in $S_{o1}$ is either a vertex before $v'$, or $v'$, or an isolated vertex after $v'$, we get that $k_1' + I_1' + 1\ge k + I_1$, so $m \le 0$. Therefore we have a solution with the same partition up to $v$ and such that $v$ is in $S_{o1}$. Applying the proof of Lemma~\ref{lem_domi3} from $v$ on, we can change it into an optimal solution respecting Lemma~\ref{lem_domi3}, with the same partition up to $v$ and such that $v$ is in $S_{o1}$.
\end{proof}

\begin{proof}[Proof of Theorem~\ref{th_algo_opti}]
By iterating Lemma~\ref{lem_iso}, we obtain an optimal solution such that every vertex of $S_1$ (from our solution) is in $S_{o1}$. Suppose that there is a vertex $v$ in $S_2 \cap S_{o1}$. By construction, it holds that there are more vertices other than remaining isolated vertices in $S_{o1}$ before this vertex $v$. It follows that moving $v$ from $S_{o1}$ to $S_{o2}$ improves the solution, which results in a contradiction. Therefore, the optimal solution $(S_{o1},S_{o2})$ is equal to the solution $(S_1,S_2)$ of Theorem~\ref{th_algo_opti}.
\end{proof}

To conclude this section, we mention the connection to the max-cut problem in the case of threshold graphs. Given a partition $(A,B)$ of the vertices of a graph $G$, consider these two sets:
\begin{enumerate}
\item The set of pairs of non-adjacent vertices, where one element is in $A$ and the other in $B$. Let us call it $C$.
\item The set of pairs of non-adjacent \newsuggestion{}{vertices} where both are in $A$ or both in $B$. Let us call it $F$ (for threshold graphs, this is a completion to PIG).
\end{enumerate}
Note that $E(G)$, $C$, and $F$ are a partition of the set of pairs of $V(G)$, hence $|E(G)| + |C| + |F| = {n \choose 2}$. Two of these four values depend only on the input graph, not on the partition. Hence, choosing a partition minimizing $|F|$ is the same as maximizing $|C|$. Now take the complement graph and consider the same partition. $C$ is clearly a cut between $A$ and $B$. So for threshold graphs, minimizing the size of the completion is equivalent to finding the maximum cut of the complement.\\
We remark that an $O(n^2)$-algorithm for calculating max-cut for cographs (and thus for threshold graphs) has been presented in \cite{Bodlaender2000}. This result, together with the fact that threshold graphs are closed under complementation, gives an alternative algorithm for PIG-completion for this class (albeit a less efficient one).

\subsection{Completion from quasi-threshold graphs}\label{sec:quasi-threshold}

\suggestion{When trying to generalize the previous results, a graph class related to threshold graphs arises as a natural continuation: quasi-threshold graphs are precisely $\{P_4, C_4\}$-free graphs.}{A natural candidate for generalizing the previous result \newsuggestion{}{is the class of}  quasi-threshold graphs. Recall that these are precisely the $\{P_4, C_4\}$-free graphs.} However, for this particular class we \suggestion{stumble upon with the fact}{encounter the problem} that a crucial result from the previous section does not hold: a minimum co-bipartite-completion is not the same as a minimum PIG-completion.
We present a counterexample for this in Figure~\ref{fig:counterexample_qt}, in which we can see that a minimum co-bipartite-completion of the graph has an induced $C_4$. 

\begin{figure}[h!]
    \centering
    \includegraphics[scale=.22]{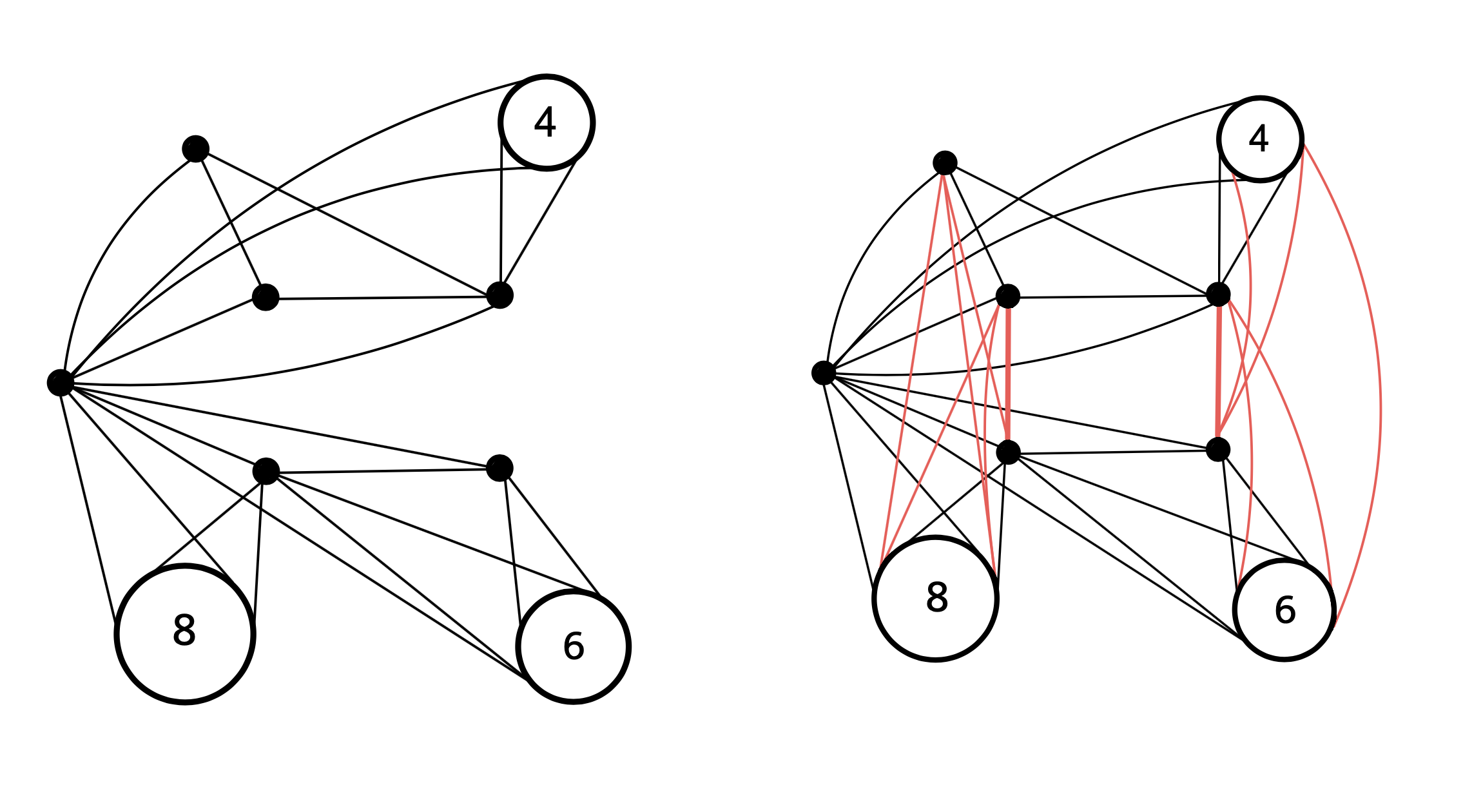}
    \caption{A quasi-threshold graph $G$ and a minimum co-bipartite-completion of $G$ that is not PIG. \suggestion{}{Numbered circles correspond to cliques of size equal to the number inside it}}
    \label{fig:counterexample_qt}
\end{figure}

Nevertheless, we still can find an algorithm that solves the minimum co-bipartite-completion problem when the input is a quasi-threshold graph. Notice that, even though this does not \suggestion{}{necessarily} provide a minimum PIG-completion for a quasi-threshold graph, it is indeed a lower bound for it. Recall that every connected quasi-threshold graph has a dominating vertex, and that any PIG-completion in particular yields a claw-free supergraph of the input graph. The fact that a minimum  co-bipartite-completion gives a lower bound for the minimum PIG-completion follows \suggestion{as a result of}{from} these two remarks.

Let \newsuggestion{$G = (V, E(G))$}{$G = (V, E_G)$} be a connected quasi-threshold graph and \newsuggestion{$T =(V, E(T))$}{$T =(V, E_T)$} be the tree rooted in $r$ that stems from its definition. 

For $v \in V$, let $n_v = |V(T_v)|$ be the number of vertices in $T_v$, the maximal subtree rooted at $v$, and  let $c_v$ be the number of children of $v$. 

For $v \in V$, let $v_1, \ldots, v_{c_v}$ be the children of $v$ in $T$, $0 \leq i \leq c_v$, $0 \leq j < n_v$. Let $X_{i,v} = \bigcup_{w \in \{v_1, \ldots, v_i\}} V(T_w)$ and $x_{i,v} = |X_{i,v}| = \sum_{1 \leq k \leq i} n_{v_k}$. We define $C(v,i,j)$ as the minimum number of edges in a co-bipartite-completion in the subgraph induced by $X_{i,v}$ such that one clique has cardinality $j$ and the other has cardinality $x_{i,v}-j$. Note that $v \not\in X_{i,v}$.

Finally, we define $D(v, j)$ as the minimum number of edges in a co-bipartite-completion of the subgraph induced by $V(T_v)$ such that one of the cliques has cardinality $j$. 

\begin{proposition} \label{prop:Dv0}
$D(v, 0) = C(v, c_v, n_v-1)$ and for $1 \leq j \leq n_v$, it holds that $D(v, j) = C(v, c_v, j-1)$.
\end{proposition}
\begin{proof}
Note that $V(T_v) = X_{c_v, v} \cup \{v\}$, and that $v$ is adjacent to every vertex in $X_{c_v, v}$, hence we can always add it to any clique without increasing the cost of the completion.
\end{proof}

It follows from those definitions that the value we are looking for is given by
\begin{equation}
\min_{0 \leq j \leq |V(G)|} D(r,j). \label{D_rj}
\end{equation}

We now show how to compute $C(v,i,j)$, for every $v$ and all possible values if $0 \leq i \leq c_v$ and $0 \leq j \leq x_{i,v}$. Once again, we assume that the children of $v$, if any, are $v_1, \ldots, v_{c_v}$.

\begin{equation}
  C(v,i,j)=\begin{cases}
    0,  \text{ if $v$ is a leaf or $i=0$}.\\
    \min\limits_{j-n_{v_i} \leq k \leq j} C(v, i-1, k) + D(v_i, j-k) + k(j-k) + (x_{i-1,v}-k)(n_{v_i}-j+k),  \text{ otherwise}.
  \end{cases}
  \label{recurrence}
\end{equation}

\begin{proposition} \label{prop:recurrence}
Equation~(\ref{recurrence}) computes $C(v,i,j)$ correctly.
\end{proposition}
\begin{proof}
For the first case, the graph is either trivial or empty.
For the second case, let $C_1$ and $C_2$ be the cliques of a completion obtained after a completion of $X_{i,v}$ such that $|C_1| = j$, and let $k = |X_{i-1,v} \cap C_1|$. 
It follows that $|V(T_{v_i}) \cap C_1| = j-k$. Notice that we need $C(v, i-1, k)$ edges to turn each of $X_{i,v} \cap C_1$ and $X_{i,v} \cap C_2$ into cliques. Similarly, we need $D(v_i, j-k)$ edges to turn $V(T_{v_i}) \cap C_1$ and $V(T_{v_i}) \cap C_2$ into cliques. Finally, we need $k(j-k)$ edges to connect $X_{i,v} \cap C_1$ and $V(T_{v_i}) \cap C_1$, and $(x_{i-1,v}-k)(n_{v_i}-j+k)$ edges to connect $X_{i,v} \cap C_2$ and $V(T_{v_i}) \cap C_2$. Since we try all possible values of $k$, we eventually find the smallest possible completion.

\end{proof}

\noindent We arrive thus at the main result of this section, which we state below:

\begin{theorem} Let $G$ be a quasi-threshold graph. There exists an $O(n^4)$ dynamic programming algorithm that computes the minimum co-bipartite-completion number for $G$.
\end{theorem}
\begin{proof} The algorithm is given by expressions (\ref{D_rj}) and (\ref{recurrence}), and \suggestion{the}{its} correctness is immediate by Propositions~\ref{prop:Dv0} and \ref{prop:recurrence}. For the complexity, notice that the algorithm proceeds in a bottom-up manner, from the leaves of $T$ up to the root, and from left to right for vertices located at the same level of the tree. For every vertex $v \in V$, every $0 \leq i \leq c_v$ and every $0 \leq j \leq x_{i,v}$, we compute the recurrence relation (\ref{recurrence}). This means that the values for expressions $C(v, i - 1, k)$ and $D(v_i, j - k)$ have already been calculated for every $j-n_{v_i} \leq k \leq j$ by the time we compute $C(v, i, j)$, so they are $O(1)$. Hence, computing $C(v, i, j)$ is $O(n)$. Since this is required for every $i$ and $j$ as defined above, we perform $O(n^3)$ operations at every vertex of $T$, except for the root. Indeed, for the root $r$, we only need to compute $C(r, c_r, j)$, for $0 \leq j \leq n$. This amounts to $O(n^2)$ operations. The complete algorithm is thus $O(n^4)$.
\end{proof}

\section{An algorithm for PIG-completion on caterpillars}\label{sec:caterpillar}

In \newsuggestion{}{the} previous sections, we studied PIG-completion within very dense subclasses of interval graphs. Another possible approach is to restrict the input to very sparse ones. Namely, caterpillars are those interval graphs that are also trees, and thus the sparsest subclass of interval graphs.
We will show that PIG-completions are very particular when the input is a caterpillar, and we will give a quadratic-time algorithm to compute them.

For a caterpillar $G$, we call $P$ the \emph{central path}, the \emph{father} of a leaf is its neighbour in $P$, and the \emph{sons} of a vertex of $P$ is the set of leaves it is adjacent to.

Let us consider a PIG-completion of $G$, and a \newsuggestion{proper interval model of it where every interval is a closed interval of length $1$}{unit interval model of it}. First, we show that such a model can be transformed into a particular one where the vertices of the central path are intervals of the form $[i,i+1]$. Later, we will describe an $O(n^2)$-time algorithm to get the best possible completion.

\begin{theorem}\label{teo:combinatorial_caterpillar}
For every caterpillar, there exists a minimum PIG-completion \newsuggestion{}{ with a unit interval model} such that \newsuggestion{the vertices of the central path are intervals}{each vertex of the central path is represented by an interval} of the form $[i,i+1]$ (where $i$ is an integer), and \newsuggestion{each of the leaves adjacent to each vertex can be represented by one of two disjoint intervals, for instance of the form $[i \times (1 + \frac{1}{k}) - 1, i \times (1 + \frac{1}{k})]$}{its sons are represented either by $[i \times (1 + \frac{1}{k}) - 1 , i \times (1 + \frac{1}{k}) ]$ or $[(i+1) \times (1 + \frac{1}{k}) - 1 , (i+1) \times (1 + \frac{1}{k}) ]$} (where $k$ is the length of the main path of the caterpillar).
\end{theorem}

\begin{proof}
First, consider all \newsuggestion{of }{}the leaves and replace each of their closed intervals of length $1$ by one that is open on the left and closed on the right. This can only decrease the number of edges in the graph. Let us now transform it into one that requires at most as many edges. To do this, we will use an interval model that is not a proper interval model, in which some intervals are points while the others are intervals that have identical size (say $1$). 
 We will later transform this model into a unit interval model that represents the same graph. 

Consider the vertices of $P$ in order $v_0, \dots, \newsuggestion{v_k}{v_{k-1}}$. For $i \in \{0, \dots,\newsuggestion{}{k-1}\}$, let $V_i$ be the set of sons of $v_i$. For every vertex $v_j$ of $P$, do the following transformation to the model:

\begin{itemize}
\item For all $i \in \{0, \dots,\newsuggestion{}{k-1}\}$,
let us denote by $V_i^\ell$ the set of sons of $v_i$ whose corresponding intervals contain the left endpoint of $v_i$, and by $V_i^r$ the set of sons of $v_i$ whose intervals contain the right endpoint of $v_i$, plus those that are represented by the same interval as $v_i$. We will refer to the vertices of $V_i^\ell$ as the \emph{left sons} of $v_i$, and the vertices of $V_i^r$ as the \emph{right sons} of $v_i$.

\item For a certain integer \newsuggestion{}{$j \in \{0, \dots,k-1\}$}, we make the following assumption, which we will call \emph{assumption $A_j$}.
\begin{itemize}
\item For every integer \newsuggestion{}{$i \in \{0, \dots,j\}$}, vertex $v_i$ is represented by the interval $[i,i+1]$. 

\item For every integer \newsuggestion{}{$i \in \{0, \dots,j-1\}$}, the vertices in $V_i^\ell$ are represented by the point $i$, and the vertices in $V_i^r$ are represented by the point $i+1$. All the leaves of vertices of $P$ from $v_j$ on are represented by intervals of size $1$ (with one side open and one closed, as explained above).
\end{itemize}
\newsuggestion{We note}{Observe} that a left son of a vertex $v_i$ for $i < j$ is adjacent to $v_i$, $v_{i-1}$ and its right sons (if $v_{i-1}$ exists), as well as maybe some vertices of $P$ from $v_j$ on and some of their sons. Similarly, a right son of a vertex $v_i$ for $i < j$ is adjacent to $v_i$, $v_{i+1}$ and its left sons (if $i<j-1$), as well as maybe some vertices of $P$ from $v_j$ on and some of their sons. Finally, a vertex $v_i$ is adjacent to its sons, $v_{i-1}$ and its right sons if they exist, and $v_{i+1}$ and its left sons, as well as maybe some vertices of $P$ from $v_j$ on and some of their sons.

Let $S_j$ be the set of vertices containing every $v_i$ \newsuggestion{for $i+1 < j$, every son of $v_i$ for $i+1 < j$,}{and their sons for $i < j-1$,} and every left son of $v_{j-1}$ (if it exists). Let us denote by $T_j$ the set of vertices containing every $v_i$ for $i > j$ and all \newsuggestion{ of}{} their sons. Notice that $v_j$, $v_{j-1}$ (if it exists) as well as the sons of $v_j$ and the right sons of $v_{j-1}$ are not in $S_j$ nor in $T_j$. 

The edges between a vertex of $S_j$ and a vertex of $T_j$ are called \emph{unimportant edges}. The other edges are \emph{important edges}. We will not count the unimportant edges, but we will show that the number of important edges never increases, and also that there are no unimportant edges in the end. Note that the set of important edges changes when $j$ augments. Let us denote $I_j$ the number of important edges for step $j$. 

Observe that, if $A_{j+1}$ is verified, then no unimportant edge can become important when we go from $I_j$ to $I_{j+1}$ since only important edges can become \newsuggestion{non-important}{unimportant}. Indeed, such an edge would have to be between $v_{j+1}$ and its leaves on the one side, and vertices of $\{v_{j-2}\} \cup V_{j-2}^r \cup V_{j-1}^\ell$ on the other side. 
But by $A_{j+1}$, the models of $\{v_{j-2}\} \cup V_{j-2}^r \cup V_{j-1}^\ell$ are in $[j-2,j-1]$, the model of $v_{j+1}$ is $[j+1,j+2]$, and its leaves are intervals of length $1$, open on one side and closed on the other, intersecting $[j+1,j+2]$, therefore the two sets do not intersect and there can be no such new important edge.

\item Let us transform the model in order to make assumption $A_{j+1}$ true. 
 
 If $v_{j+1}$ has its left extremity to the left of $v_j$, then do a symmetry of the models of $T_j$ by point $j + \frac{1}{2}$.

Now $v_{j+1}$ does not have its left extremity to the left of $v_j$. We translate the models of $T_j$ to the right, so that $v_{j+1}$ corresponds to $[j+1,j+2]$. Let \newsuggestion{$K$}{$C$} be the clique formed by all\newsuggestion{ of}{} the vertices of $T_j$ whose models intersect with point $j+1$ after the transformations. Note that before the transformations, the vertices of \newsuggestion{$K$}{$C$}  were all intersecting a singular point in $[j,j+1]$. Indeed, the translation was of at most one to the right, and the symmetry did not alter this. The vertices of \newsuggestion{$K$}{$C$}  can be partitioned into the set \newsuggestion{$K_1$}{$C_1$}  of those that initially intersected $j$ and the set \newsuggestion{$K_2$}{$C_2$}  of those that did not. Since every vertex of \newsuggestion{$K$}{$C$}  is represented by an interval of length $1$, either a closed one or one closed on one side and open on the other, vertices in \newsuggestion{$K_2$}{$C_2$}  initially intersected $j+1$.
Now we switch $\min ( |V_j^r|, |C_1|)$ leaves from $V_j^r$ to $V_j^\ell$, replace leaves in $V_j^\ell$ by point $j$ and those in $V_j^r$ by point $j+1$.
Although we changed the graph as described previously and this would change the value of the sizes of the sets, in our equations the values of $|V_j^r|$, $|V_j^\ell|$, \newsuggestion{}{$|C_1|$, $|C_2|$, and $|C|$} will remain as before the transformations, so that there is no ambiguity. 

It is easy to see that $A_{j+1}$ is now verified. However, we must show that the number of important edges did not increase. We will compare here the number of important edges before the symmetry/translation, with the number of important edges after $I_j$ is changed to $I_{j+1}$. The replacement of intervals by a point contained in the interval may of course only remove edges, and not add edges. The translation and the symmetry do not alter edges between vertices of $T_j$, nor edges between vertices that are not in $T_j$. The only edges that can be added are thus between vertices that are in $T_j$ and vertices that are not in $T_j$, plus edges between vertices that were switched from $V_j^r$ to $V_j^\ell$. Among those edges, we only need consider those that are in $I_{j+1}$. 

Let us first consider the edges that may be added between vertices in $S_{j+1}$ and $T_j$. Note that vertices in $S_{j+1}$ are represented in $[0,j]$. For one such edge to be in $I_{j+1}$, it would have to be between a vertex of $S_{j+1}$ on one side, and $v_{j+1}$ or one of its leaves on the other side. But it follows from a previous discussion that the model of $v_{j+1}$ is $[j+1,j+2]$, and its leaves are intervals of length $1$, open on one side and closed on the other, intersecting $[j+1,j+2]$. Therefore the two sets do not intersect and there can be no such new important edge.

Thus, we only need to consider the edges between vertices in the two cliques represented by points $j$ and $j+1$. Moreover, we only need to consider those edges between $T_j$ on the one hand and the leaves of $v_j$ on the other hand, plus the edges between the leaves of $v_j$ that are transferred on the one hand, and the other leaves of $v_j$ to which they are adjacent to on the other hand. Let $V_\ell = \{v_{j-1}\} \cup V_{j-1}^r \cup V_j^\ell$. 

After the transformation, the number of edges to consider is $|V
^\ell| \times \min ( |V_j^r|, |C_1|) + |C| \times (|V_j^r| - \min ( |V_j^r|, |C_1|))$. Since (before the transformation) the vertices of $V_\ell$ are adjacent to those of \newsuggestion{$K_1$}{$C_1$} and those of $V_j^r$ are adjacent to those of \newsuggestion{$K_2$}{$C_2$} plus the other vertices of $V_j^r$, it follows that the number of edges considered before the transformation is at least $|V^\ell| \times |C_1| + |C_2| \times (|V_j^r| - \min ( |V_j^r|, |C_1|)) + (|C_2| + |V_j^r| - \min ( |V_j^r|, |C_1|)) \times \min ( |V_j^r|, |C_1|)$.

\begin{itemize}
    \item Case $|V_j^r| \ge |C_1|$. In this case we consider $|V^\ell| \times |C_1| + |C| \times (|V_j^r| - |C_1|)$ edges after the transformation, and at least $|V^\ell| \times |C_1| + |C_2| \times (|V_j^r| - |C_1|) + (|C_2| + |V_j^r| - |C_1|) \times |C_1|$, therefore there are at least $|C_2| \times |C_1| \ge 0$ more edges considered before the transformation than after, and the number of important edges does not increase.

    \item Case $|V_j^r| < |C_1|$. In this case we consider $|V^\ell| \times |V_j^r|$ edges after the transformation, and at least $|V^\ell| \times |C_1| + |C_2| \times |V_j^r|$ before the transformation, therefore there are at least $|C_2| \times |V_j^r| \ge 0$ more edges considered before the transformation than after, and the number of important edges does not increase.
\end{itemize}

\end{itemize}

Note that we can assume that $A_0$ (that is precisely that $v_0$ is represented by $[0,1]$) is verified in the beginning and holds, up to translation of the model. Then, after applying the previous transformations for every $j \in \{0,\dots,\newsuggestion{}{k-2}\}$, we know that $A_{\newsuggestion{}{k-1}}$ is true and that the number of edges has not increased (for $A_{\newsuggestion{}{k-1}}$, $T_{\newsuggestion{}{k-1}} = \emptyset$, and thus every edge is important). Now we can replace every leaf in $V_{\newsuggestion{}{k-1}}^\ell$ by the point $\newsuggestion{}{k-1}$ and every leaf in $V_{\newsuggestion{}{k-2}}^r$ by the point $\newsuggestion{}{k}$ to obtain a model where:

\begin{itemize}
    \item For every $i \in \{0,\dots,\newsuggestion{}{k-1}\}$, the vertex $v_i$ is represented by the interval $[i,i+1]$. 

    \item For every $i \in \{0,\dots,\newsuggestion{}{k-1}\}$, the vertices in $V_i^\ell$ are represented by the point $i$, and the vertices in $V_i^r$ are represented by the point $i+1$.
\end{itemize}

For each $i \in \{0,\dots,k\}$, let us denote by $W_i$ the set of vertices represented by the integer $i$.
Now we keep the model of every $v_i$ and replace each vertex in $W_i$ with the interval \newsuggestion{}{$[i \times (1 + \frac{1}{k})-1, i \times (1 + \frac{1}{k})]$}. These intervals, which are closed of length $1$, indeed do not intersect unless they are equal, and keep the same intersections with the $v_i$'s as their previous models. 
Thus we get a unit interval model of the same graph as obtained with the previous algorithm, which therefore has no fewer edges than the initial one.
\end{proof}

An interesting question that arises is whether this suffices to characterize all inclusion-wise minimal completions or not. In Figure~\ref{fig:counterexample_minimal} we can see depicted a caterpillar graph $G$, represented by the black edges, and a supergraph $H$ obtained by adding the red edges that is indeed an inclusion-wise minimal PIG-completion of $G$. However, $H$ is not obtained by splitting into two parts each set of leaves adjacent to a same vertex of the spine.

\begin{figure}
    \centering
    \includegraphics[scale=.3]{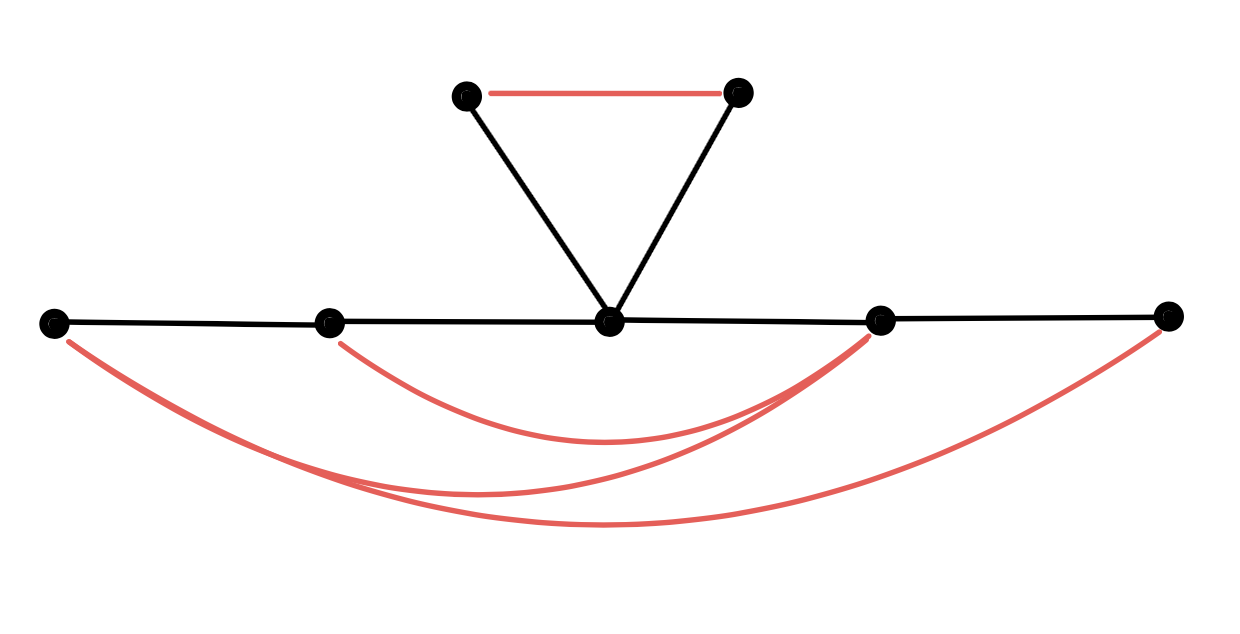}
    \caption{A caterpillar graph (black edges) and an inclusion-wise minimal PIG-completion (red edges) for which the characterization given in Theorem~\ref{teo:combinatorial_caterpillar} does not hold.}
    \label{fig:counterexample_minimal}
\end{figure}

Now that we know that there exists a minimum PIG-completion as the one stated in Theorem~\ref{teo:combinatorial_caterpillar}, we design an algorithm that will always output such a solution. 

\begin{theorem}
There exists an algorithm that outputs a minimum PIG-completion of a caterpillar in $O(n^2)$-time.
\end{theorem}
\begin{proof}
We keep using the notation defined in the proof of the theorem above.
The only thing that remains to be determined is the size of $V_i^\ell$ (or equivalently that of $V_i^r$) for every $i \in \{1,\dots,k\}$. We will determine them via dynamic programming. More precisely, for every $i \in \{1,\dots,k\}$, and for every $j \in \{0, \dots, |V_i|\}$, we compute the optimal number $N_{i,j}$ of edges added to complete the caterpillar induced by $\{v_i, v_{i+1}, \dots, v_{k}\} \cup W_{i+1} \cup W_{i+2} \cup \dots \cup W_{k+1}$ into a proper interval graph, given that $|V_i^\ell| = j$. We use the following formula to compute it: 

\begin{itemize}
    \item For $j \in \{0, \dots, V_k\}$, $N_{k,j} = \binom{|V_k| - j}{2}$;
    
    \item For $i \in \{1,\dots,k-1\}$ and $j \in \{0, \dots, V_i\}$, $N_{i,j} = \min_{j' \in \{0, \dots, |V_{i+1}|\}}\left (\binom{|V_k| - j + j' + 1}{2} + N_{i+1,j'} \right )$.
\end{itemize}

Now, the best PIG-completion of the whole caterpillar graph uses exactly \[\min\limits_{j \in \{0, \dots, |V_{0}|\}}\left (\binom{j}{2} + N_{1,j} \right )\] edges. One can get such a completion from the choices of $j$ obtained to get the minimal values.

Since $\sum_{i \in \{0, \dots, k\}}|V_i| \le n$, we know that there are at most $n$ different values of $N_{i,j}$ to be computed. Each one is computed using a minimum over at most $n$ different expressions. Therefore the complexity of this algorithm is in $O(n^2)$.
\end{proof}




\section{Conclusions and future work}\label{sec:conclusions}

In this work, we study the $\Pi$-completion problem when $\Pi$ is the class of proper interval graphs. Given that the problem is NP-complete in general graphs, we focus \rv{our analysis} 
\rv{on} 
the case in which the input graph lies in some particular subclasses of chordal graphs. We prove that the problem remains hard in split graphs, whereof we conclude the same for chordal graphs. We present efficient algorithms for \suggestion{}{PIG-completion for both} threshold \suggestion{}{graphs} and caterpillar graphs, and an efficient algorithm for co-bipartite-completion for quasi-threshold graphs. \\
A future line of work is to continue \rv{studying} the PIG-completion problem in other subclasses of chordal graphs that have bounded cliquewidth to obtain practical polynomial-time algorithms. This may lead to finding common properties that could be useful when it comes to designing efficient algorithms and heuristics to solve the problem within other chordal subclasses.\\
On the other hand, given that all the graph classes for which we give a polynomial-time algorithm are also subclasses of interval graphs, and that interval and proper interval graphs are very closely related, it raises as a natural question whether the PIG-completion problem can be solved in polynomial time when the input graph already belongs to this particular class.\\
In addition, since we studied the PIG-completion problem within caterpillars motivated by the fact that these graphs are precisely those interval graphs that are also trees, this gives way to an analogous question regarding the complexity of the PIG-completion problem when the input is a tree. Since PIG-completion is expressable in $MSOL_2$, the problem is polynomial-time solvable on trees. However, given that our current algorithm strongly relies on an interval model of the graph, the problem on trees in general should be addressed using new techniques. \\
An interesting question arises also in the relation with the max-cut problem: for which other classes besides threshold graphs does it hold that PIG-completion is equivalent to co-bipartite-completion? For these possible classes, an algorithm for max-cut in the complement would also solve the PIG-completion problem, by the same argument given in Section \ref{sec:threshold}.\\
Another possible continuation for this work may be to study the \emph{PIG-deletion problem}, i.e., the removal of a set of edges $F$ from an input graph $G =(V,E)$, so that the resulting subgraph $H = (V, E \setminus F)$ is a proper interval graph. It is known that this problem is hard for general graphs~\cite{Goldberg1995}, but it would be interesting to investigate whether efficient algorithms could also be devised for this problem restricted to the aforementioned subclasses of chordal graphs.\\




\section*{Acknowledgements}

Partially supported by Programa Regional MATHAMSUD MATH190013, by CNPq grant 311679/2018-8, Argentina PIP 2021-2023 20020190200124BA, Argentina UBACyT 20020170100495BA.

\bibliographystyle{plain} 
\bibliography{biblio} 

\begin{thebibliography}{10}

\bibitem{Bliznets2015}
I.~Bliznets, F.~V. Fomin, M.~Pilipczuk, and M.~Pilipczuk.
\newblock A {S}ubexponential {P}arameterized {A}lgorithm for {P}roper
  {I}nterval {C}ompletion.
\newblock {\em SIAM Journal on Discrete Mathematics}, 29(4):1961--1987, 2015.

\bibitem{Bodlaender2000}
H.~Bodlaender and K.~Jansen.
\newblock On the complexity of the {M}aximum {C}ut {P}roblem.
\newblock {\em Nordic J. of Computing}, 7(1):14–31, 2000.

\bibitem{Bondy}
J.~A. Bondy and U.~S.~R. Murty.
\newblock {\em {G}raph {T}heory}.
\newblock Springer Publishing Company, Incorporated, 1st edition, 2008.

\bibitem{Courcelle00lineartime}
B.~Courcelle, J.~A. Makowsky, and U.~Rotics.
\newblock Linear time solvable optimization problems on graphs of bounded
  clique-width, 2000.

\bibitem{CrespelleT13}
C.~Crespelle and I.~Todinca.
\newblock An ${O}(n^2)$-time algorithm for the minimal interval completion
  problem.
\newblock {\em Theor. Comput. Sci.}, 494:75--85, 2013.

\bibitem{Diaz1991}
J.~Diaz, A.~M. Gibbons, M.~S. Paterson, and J.~Toran.
\newblock The {MINSUMCUT} problem.
\newblock In {\em Algorithms and Data Structures}, volume 519 of {\em Lecture
  Notes in Computer Science}, pages 65--79. Springer, Berlin, 1991.

\bibitem{GAREY}
M.~R. Garey and D.~S. Johnson.
\newblock {\em Computers and Intractability: A Guide to the Theory of
  {NP}-Completeness}.
\newblock W. H. Freeman and Company, N.Y., 1979.

\bibitem{Goldberg1995}
P.~Goldberg, M.~C. Golumbic, H.~Kaplan, and R.~Shamir.
\newblock Four strikes against physical mapping of {DNA}.
\newblock {\em J. Comput. Bio}, 2(1):139--152, 1995.

\bibitem{HeggernesSTV07}
P.~Heggernes, K.~Suchan, I.~Todinca, and Y.~Villanger.
\newblock {C}haracterizing {M}inimal {I}nterval {C}ompletions.
\newblock In Wolfgang Thomas and Pascal Weil, editors, {\em {STACS} 2007, 24th
  Annual Symposium on Theoretical Aspects of Computer Science, Aachen, Germany,
  February 22-24, 2007, Proceedings}, volume 4393 of {\em Lecture Notes in
  Computer Science}, pages 236--247. Springer, 2007.

\bibitem{K-S-T-interval}
H.~Kaplan, R.~Shamir, and R.~Tarjan.
\newblock Tractability of parameterized completion problems on chordal,
  strongly chordal and proper interval graphs.
\newblock 28(5):1906--1922, 1999.

\bibitem{Kashiwabara1979}
T.~Kashiwabara and T.~Fujisawa.
\newblock An {NP}-complete problem on interval graphs.
\newblock {\em IEEE Symp. Of Circuits and Systems}, page 82–83, 1979.

\bibitem{Lekkeikerker}
C.~Lekkerkerker and J.~Boland.
\newblock Representation of a finite graph by a set of intervals on the real
  line.
\newblock {\em Fundamenta Mathematicae}, 51:45--64, 1962.

\bibitem{Natanzon2001}
A.~Natanzon, R.~Shamir, and R.~Sharan.
\newblock Complexity classification of some edge modification problems.
\newblock {\em Discrete Appl. Math.}, 113:109–128, 2001.

\bibitem{RapaportST08}
I.~Rapaport, K.~Suchan, and I.~Todinca.
\newblock Minimal proper interval completions.
\newblock {\em Inf. Process. Lett.}, 106(5):195--202, 2008.

\bibitem{roberts_phd_th}
F.~S. Roberts.
\newblock {\em Representations of indifference relations}.
\newblock PhD thesis, 1968.

\bibitem{roberts1969indifference}
F.~S. Roberts.
\newblock Indifference graphs. proof techniques in graph theory.
\newblock In {\em Proceedings of the Second Ann Arbor Graph Conference,
  Academic Press, New York}, 1969.

\bibitem{Benzer1959}
B.~Seymour.
\newblock On the topology of the genetic fine structure.
\newblock {\em In Proceedings of the National Academy of Sciences of the United
  States of America}, 45:1607–1620, 1959.

\bibitem{PengChen06}
P.~Sheng{-}Lung and C.~Chi{-}Kang.
\newblock On the interval completion of chordal graphs.
\newblock {\em Discret. Appl. Math.}, 154(6):1003--1010, 2006.

\bibitem{Yannakakis1981}
M.~Yannakakis.
\newblock Computing the minimum fill-in is {NP}-complete.
\newblock {\em SIAM J. Alg. Disc. Math.}, 2(1):77–79, 1981.

\bibitem{Yuan1998}
J.~Yuan, Y.~Lin, Y.~Liu, and S.~Wang.
\newblock {NP}-completeness of the profile problem and the fill-in problem on
  cobipartite graphs.
\newblock {\em J. Math. Study}, 31:239–243, 1998.

\end{thebibliography}


\end{document}